\documentclass[sn-mathphys-num]{sn-jnl}% Math and Physical Sciences Numbered Reference Style 
\usepackage{graphicx}%
\usepackage{multirow}%
\usepackage{amsmath,amssymb,amsfonts}%
\usepackage{amsthm}%
\usepackage{mathrsfs}%
\usepackage{enumitem}
\usepackage[title]{appendix}%
\usepackage{xcolor}%
\usepackage{textcomp}%
\usepackage{manyfoot}%
\usepackage{booktabs}%
\usepackage{algorithm}%
\usepackage{algorithmicx}%
\usepackage{algpseudocode}%
\usepackage{listings}%
\usepackage{cleveref}
\usepackage{geometry}
\theoremstyle{thmstyleone}%
\newtheorem{theorem}{Theorem}%  meant for continuous numbers
\newtheorem{lemma}{Lemma}
\newtheorem{corollary}{Corollary}
\theoremstyle{thmstyletwo}%
\newtheorem{example}{Example}%
\newtheorem{remark}{Remark}%
\newtheorem{conjecture}{Conjecture}

\theoremstyle{thmstylethree}%

\crefname{equation}{}{}
\crefname{theorem}{Theorem}{Theorem}
\crefname{corollary}{Corollary}{Corollary}
\crefname{lemma}{Lemma}{Lemma}
\crefname{remark}{Remark}{Remark}
\crefname{definition}{Definition}{Definition}
\crefname{conjecture}{Conjecture}{Conjecture}
\crefname{example}{Example}{Example}
\crefname{section}{Section}{Section}

\begin{document}

\title[Article Title]{The dual codes of two families of BCH codes\footnotemark[2]}


\author[1]{\fnm{Haojie} \sur{Xu}}\email{xuhaojiechn@163.com}

\author*[1]{\fnm{Xia} \sur{Wu}}\email{wuxia80@seu.edu.cn}

%	\equalcont{These authors contributed equally to this work.}

\author[1]{\fnm{Wei} \sur{Lu}}\email{luwei1010@seu.edu.cn}
%	\equalcont{These authors contributed equally to this work.}

\author[2]{\fnm{Xiwang} \sur{Cao}}\email{xwcao@nuaa.edu.cn}

\affil*[1]{\orgdiv{School of Mathematics}, \orgname{Southeast University}, \orgaddress{\city{Nanjing}, \postcode{210096}, \country{China}}}

\affil[2]{\orgdiv{Department of Math}, \orgname{Nanjing University of Aeronautics and Astronautics}, \orgaddress{\city{Nanjing}, \postcode{211100}, \country{China}}}

\footnotetext[2]{Supported by NSFC (Nos. 12371035, 12171241)}


\abstract{In this paper, we present an infinite family of MDS codes over $\mathbb{F}_{2^s}$ and two infinite families of almost MDS codes over $\mathbb{F}_{p^s}$ for any prime $p$, by investigating the parameters of the dual codes of two families of BCH codes. Notably, these almost MDS codes include two infinite families of near MDS codes over $\mathbb{F}_{3^s}$, resolving a conjecture posed by Geng et al. in 2022. Furthermore, we demonstrate that both of these almost AMDS codes and their dual codes hold infinite families of $3$-designs over \(\mathbb{F}_{p^s}\) for any prime $p$. Additionally, we study the subfield subcodes of these families of MDS and near MDS codes, and provide several binary, ternary, and quaternary codes with best known parameters.}




\maketitle

\section{Introduction}\label{sec1}

Let $q$ be the power of a prime $p$, $\mathbb{F}_q$ be the finite field with $q$ elements and $\mathbb{F}_q^\ast=\mathbb{F}_q\setminus\{0\}$. An $[n,k]$ linear code $\mathcal{C}$ over $\mathbb{F}_q$ is a vector subspace of $\mathbb{F}_q^n$ with dimension $k$. The \textit{minimum distance} $d$ of $\mathcal{C}$ is represented by $d=\min\{\mathrm{wt}(\mathbf{c}):\mathbf{c}\in\mathcal{C}\setminus \{\mathbf{0}\}\}$, where $\mathrm{wt}(\mathbf{c})$ is the number of nonzero coordinates of $\mathbf{c}$ and called the \textit{Hamming weight} of $\mathbf{c}$. Denote by $A_i$ the number of codewords of Hamming weight $i$ in $\mathcal{C}$, where $0\le i \le n$. The sequence $(A_i)_{i=0}^n$ is referred to as the \textit{weight distribution} of $\mathcal{C}$, and the polynomial $\sum_{i=0}^{n}A_iz^i$ is called the \textit{weight enumerator} of $\mathcal{C}$. A code $\mathcal{C}$ is referred to as \textit{$t$-weight code} if $|\{1\le i\le n : A_i \ne 0\}|=t$. The dual code $\mathcal{C}^\perp$ of $\mathcal{C}$ is defined as 
\begin{equation*}
	\mathcal{C}^\perp=\left\{\mathbf{v}\in\mathbb{F}_q^n:\mathbf{c}\cdot\mathbf{v}=0,\ \forall\;\mathbf{c}\in\mathcal{C}\right\},
\end{equation*}
where $\mathbf{c}\cdot\mathbf{v}$ is the inner product of $\mathbf{c}$ and $\mathbf{v}$. Analogously, let $d^\perp$ denote the minimum distance of $\mathcal{C}^\perp$, and $(A_i^\perp)_{i=0}^n$ denote the weight distribution of $\mathcal{C}^\perp$. The \textit{MacWilliams equations} relating $(A_i)_{i=0}^n$ and $(A_i^\perp)_{i=0}^n$ \cite[Lemma 2.9]{MacWilliams1963TheoremDistributionWeight} are given by
\begin{equation}\label{sec1.1 equ1}
	\sum_{i=0}^{n-j}\begin{pmatrix}n-i\\j\end{pmatrix}A_i=q^{k-j}\sum_{i=0}^j\begin{pmatrix}n-i\\n-j\end{pmatrix}A_i^{\perp}\quad\mathrm{for}\ 0\le j \le n.
\end{equation}
It is desirable to design linear codes with the largest possible rate $\frac{k}{n}$ and minimum distance $d$ in coding theory. Nevertheless, there are some tradeoffs among $n$, $k$ and $d$. The Singleton bound indicates that $d\le n-k+1$ for a linear code with parameters $[n,k,d]$. Linear codes that meet the Singleton bound are called maximum distance separable, or MDS for short. Note that if $\mathcal{C}$ is MDS so is the dual code $\mathcal{C}^\perp$. An $[n,k,n-k]$ code is said to be almost MDS (AMDS for short) \cite{Boer1996AlmostMDSCodes}. Unlike MDS code, the dual code of an AMDS code need not be AMDS. Furthermore, $\mathcal{C}$ is said to be near MDS (NMDS for short) if both $\mathcal{C}$ and its dual code $\mathcal{C}^\perp$ are AMDS codes \cite{Dodunekov1995NearMDSCodes}. MDS, NMDS and AMDS codes play vital roles in communications, data storage, combinatorial theory, and secret sharing \cite{Ding2020InfiniteFamiliesMDS,Tang2021InfiniteFamilyLinear,Heng2022ConstructionsMDSNear,Xu2022InfiniteFamiliedDesigns,Yan2022InfiniteFamiliesLinear,Li2024MDSCodesConjecture,Heng2023NewInfiniteFamilies,Dodunekova1997AlmostMDSNearMDSCodes,Xiang2022SomeDesignsBCH,Zhou2009SecretSharingScheme}.

In 2020, Ding and Tang \cite{Ding2020InfiniteFamiliesMDS} introduced two infinite families of NMDS codes holding 2-designs or 3-designs, derived from the narrow-sense BCH codes $\mathcal{C}_{(2^s,2^s+1,3,1)}$ and $\mathcal{C}_{(3^s,3^s+1,3,1)}$. Shortly thereafter, Tang and Ding \cite{Tang2021InfiniteFamilyLinear} demonstrated that the BCH code $\mathcal{C}_{(2^s,2^s+1,4,1)}$ with odd $s$ is an NMDS code holding 4-designs, which settled a long-standing open problem. Subsequently, Yan and Zhou \cite{Yan2022InfiniteFamiliesLinear} further investigated the codewords of different weights in the BCH code $\mathcal{C}_{(2^s,2^s+1,4,1)}$ and discovered additional 3-designs and 4-designs. In 2022, Geng et al. \cite{Geng2022ClassAlmostMDS} proved that the BCH code $\mathcal{C}_{(3^s,3^s+1,3,4)}$ is AMDS when $s$ is odd, and proposed the following conjecture:
\begin{conjecture}\cite[Conjecture 3.6]{Geng2022ClassAlmostMDS}
	\label{sec1 conj1}
	Let $q=3^s$, where $s$ is odd. Then $\mathcal{C}^\perp_{(q,q+1,3,4)}$ is an AMDS code with parameters $[q+1,4,q-3]$.
\end{conjecture}
\noindent In 2023, Xu et al. \cite{Xu2023InfiniteFamiliesMDS} presented the sufficient and necessary conditions for the minimum distance of the BCH codes $\mathcal{C}_{(p^s,p^s+1,3,h)}$ to be 3 and 4, and provided several infinite families of AMDS codes. Very recently, Qiang et al. \cite{Qiang2024CharacterizationAlmost} presented a characterization for the AMDS code to be NMDS and resolved \cref{sec1 conj1} through their result.

Motivated the works of \cite{Ding2020InfiniteFamiliesMDS,Geng2022ClassAlmostMDS,Xu2023InfiniteFamiliesMDS}, this paper investigates the parameters of the dual codes $\mathcal{C}^\perp_{(p^s,p^s+1,3,\frac{p^s-p^i}{2})}$ and $\mathcal{C}^\perp_{(p^s,p^s+1,3,\frac{p^i-1}{2})}$ of two classes BCH codes. This leads to the discovery of an infinite family of MDS codes over $\mathbb{F}_{2^s}$ and two infinite families of almost MDS codes over any finite field $\mathbb{F}_{p^s}$. Specifically, the two classes of almost MDS code encompass two infinite families of near MDS codes over $\mathbb{F}_{3^s}$. Notably, \cref{sec1 conj1} is also verified by a specific subclass of the near MDS code $\mathcal{C}^\perp_{(3^s,3^s+1,3,\frac{3^i-1}{2})}$, which offers a different yet more concise method than that presented in \cite{Qiang2024CharacterizationAlmost}. Additionally, the subfield subcodes of these families of MDS and near MDS codes are also studied. Through these subfield subcodes, several binary, ternary, and quaternary codes with best known parameters are provided.

Furthermore, we establish that the minimum weight codewords of our almost MDS codes and their duals hold infinite families of $3$-designs over any field $\mathbb{F}_{p^s}$, respectively. Using the MacWilliams identities and the parameters of these $3$-designs, we determine the weight distributions of the duals of these almost MDS codes. Building upon a generalization of the Assmus-Mattson theorem \cite[Theorem 5.3]{Tang2020CodesDifferentiallyUniform}, we prove that the codewords of weight 5 in these near MDS codes over $\mathbb{F}_{3^s}$ also support a $3$-design. Moreover, our results on near MDS codes holding $3$-designs include relevant findings in \cite[Theorem 21]{Ding2020InfiniteFamiliesMDS}.

\section{Preliminaries}\label{sec2}
\subsection{Cyclic codes and BCH codes}\label{sec2.1}

An $[n,k]$ code $\mathcal{C}$ over $\mathbb{F}_q$ is called \textit{cyclic} if $\mathbf{c}=(c_0,c_1,\dots,c_{n-2},c_{n-1})\in\mathcal{C}$ implies $(c_{n-1},c_0,c_1,\dots,c_{n-2})\in\mathcal{C}$. Throughout this subsection, we assume that $\gcd(n,q)=1$. The residue class ring $\mathbb{F}_q[x]/(x^n-1)$ is isomorphic to $\mathbb{F}_q^n$ as a vector space over $\mathbb{F}_q$, where the isomorphism is given by $(c_0,c_1,\cdots,c_{n-1})\leftrightarrow c_0+c_1x+\cdots+c_{n-1}x^{n-1}$. Because of this isomorphism, any linear code $\mathcal{C}$ of length $n$ over $\mathbb{F}_q$ corresponds to a subset of $\mathbb{F}_q[x]/(x^n-1)$. Then we can interpret $\mathcal{C}$ as a subset of $\mathbb{F}_q[x]/(x^n-1)$. It is easily checked that the linear code $\mathcal{C}$ is cyclic if and only if $\mathcal{C}$ is an ideal of $\mathbb{F}_q[x]/(x^n-1)$.

Since every ideal of $\mathbb{F}_q[x]/(x^n-1)$ is principal, every cyclic code $\mathcal{C}$ is generated by the unique monic polynomial $g(x)\in\mathbb{F}_q[x]$ of the lowest degree, i.e., $\mathcal{C}=\left<g(x)\right>$. $g(x)$ is called the \textit{generator polynomial} of $\mathcal{C}$ and $h(x)=(x^n-1)/g(x)$ is called the \textit{parity-check polynomial} of $\mathcal{C}$. It should be noticed that the generator polynomial $g(x)$ is a factor of $x^n-1$, thus we have to study the canonical factorization of $x^n-1$ over $\mathbb{F}_q$ to handle $g(x)$. With this intention, we need to introduce $q$-cyclotomic cosets modulo $n$ \cite{Ding2014CodesDifferenceSetsa}.

Let $\mathbb{Z}_n$ denote the set $\{0,1,2,\dots,n-1\}$ and let $s<n$ be a nonnegative integer. The \textit{$q$-cyclotomic coset of $s$ modulo $n$} is given by 
\begin{equation}\label{sec2.1 equ1}
	C_s=\{s,sq,sq^2,\ldots,sq^{\ell_s-1}\}\bmod n\subseteq\mathbb{Z}_n,
\end{equation}
where $\ell_s$ is the smallest positive integer such that $s \equiv sq^{\ell_s} \pmod n$, and is the size of the $q$-cyclotomic coset. The following result will be useful for calculating $\ell_s$.
\begin{theorem}\cite[Theorem 4.1.4]{Huffman2003FundamentalErrorCodes}
	\label{sec2.1 thm1}
	The size $\ell_s$ of each $q$-cyclotomic coset $C_s$ is a divisor of the size $\ell_1$ of $C_1$.
\end{theorem}
The smallest integer in $C_s$ is called the \textit{coset leader} of $C_s$. Let $\Gamma_{(n,q)}$ be the set of all the coset leaders. We have then $C_s \cap C_t = \emptyset$ for any two distinct elements $s$ and $t$ in $\Gamma_{(n,q)}$, and
\begin{equation}\label{sec2.1 equ2}
	\bigcup_{s\in\Gamma_{(n,q)}}C_s=\mathbb{Z}_n.
\end{equation}
Therefore, the distinct $q$-cyclotomic cosets modulo $n$ partition $\mathbb{Z}_n$. 

Let $m$ be the multiplicative order of $q$ modulo $n$. Let $\alpha$ be a generator of $\mathbb{F}_{q^m}$ and $\beta=\alpha^{(q^m-1)/n}$. Then $\beta$ is a primitive $n$-th root of unity in $\mathbb{F}_{q^m}$. Denote the minimal	polynomial over $\mathbb{F}_q$ of $\beta^s$ by $\mathrm{M}_{\beta^s}(x)$. And this polynomial is given by
\begin{equation}\label{sec2.1 equ3}
	\mathrm{M}_{\beta^s}(x)=\prod\limits_{i\in C_s}(x-\beta^i)\in\mathbb{F}_q[x],
\end{equation}
which is irreducible over $\mathbb{F}_q$. It then follows from \cref{sec2.1 equ2} that
\begin{equation*}
	x^n-1=\prod_{s\in\Gamma_{(n,q)}}\mathrm{M}_{\beta^s}(x),
\end{equation*}
which is the canonical factorization of $x^n - 1$ over $\mathbb{F}_q$.

BCH codes are a crucial class of cyclic codes due to its exceptional error-correcting capabilities and simple encoding and decoding algorithms. Let $h$ be a nonnegative integer and $m$ be the multiplicative order of $q$ modulo $n$. Suppose that $\beta\in\mathbb{F}_{q^m}$ is a primitive $n$-th root of unity and $\delta$ is an integer with $2\le\delta\le n$. Then a \textit{BCH code} over $\mathbb{F}_q$ of length $n$ and \textit{designed distance} $\delta$, denoted by $\mathcal{C}_{(q,n,\delta,h)}$, is a cyclic code generated by 
\begin{equation*}
	g_{(q,n,\delta,h)}=\mathrm{lcm}\left(\mathrm{M}_{\beta^h}(x),\mathrm{M}_{\beta^{h+1}}(x),\dots,\mathrm{M}_{\beta^{h+\delta-2}}(x)\right),
\end{equation*}
where the least common multiple is computed over $\mathbb{F}_q$ and $\mathrm{M}_{\beta^s}(x)$ is represented in \cref{sec2.1 equ3}. If $h=1$, the BCH code $\mathcal{C}_{(q,n,\delta,h)}$ is called \textit{narrow-sense BCH code}. If $n=q^m-1$, the corresponding BCH codes are said to be \textit{primitive}.  According to the BCH bound, the minimum distance of the BCH code $\mathcal{C}_{(q,n,\delta,h)}$ is at least $\delta$.

\subsection{Combinatorial designs from linear codes}
Let $\mathcal{P}$ be a set of $n$ elements and $\mathcal{B}$ be a collection of $k$-subsets of $\mathcal{P}$. Let $t$ be an integer with $1\le t \le k$. If every $t$-subset of $\mathcal{P}$ is contained in exactly $\lambda$ elements of $\mathcal{B}$, then the incidence structure $\mathbb{D}=(\mathcal{P},\mathcal{B})$ is called a \textit{$t$-$(n,k,\lambda)$ design}, or simply \textit{$t$-design}. The elements of $\mathcal{P}$ and $\mathcal{B}$ are referred to as \textit{points} and \textit{blocks}, respectively. Denote the number of blocks in $\mathcal{B}$ by $b$. In a $t$-$(n,k,\lambda)$ design we have \cite[Chapter 4]{Ding2018DesignsLinearCodes}
\begin{equation}
	\label{sec1.2 equ1}
	\binom{n}{t}\lambda=b\binom{k}{t}.
\end{equation}
Let $\binom{\mathcal{P}}{k}$ denote the set of all $k$-subsets of $\mathcal{P}$. It is obvious that $(\mathcal{P},\binom{\mathcal{P}}{k})$ is a $k$-$(n,k,1)$ design. Such design is called a \textit{complete design}, which is trivial. A \textit{Steiner system}, denoted by $S(t,k,n)$ is a $t$-$(n,k,\lambda)$ design with $t\ge2$ and $\lambda=1$. A $t$-design $(\mathcal{P},\mathcal{B})$ is referred to as \textit{simple} if $\mathcal{B}$ does not contain any repeated blocks. In this paper, we consider only simple $t$-design with $t<k<n$.

A common technique for constructing $t$-design involves utilizing linear codes. Let $\mathcal{C}$ be a linear code of length $n$ and $\mathbf{c}=(c_{0},c_{1},\ldots,c_{n-1})$ be a codeword in $\mathcal{C}$. Let $\mathcal{P}(\mathcal{C})=\{0,1,\dots,n-1\}$ be the set of the coordinates positions of $\mathcal{C}$. The \textit{support} of $\mathbf{c}$ is defined by 
\begin{equation*}
	\mathrm{suppt}(\mathbf{c})=\{0\leq i\leq n-1:c_{i}\neq0\}.
\end{equation*}
Two different codewrods of weight $k$ may have the same support. For each $k$ with $A_k\ne 0$, let 
\begin{equation*}
	\mathcal{B}_k(\mathcal{C})=\frac{1}{q-1}\{\{\mathrm{suppt}(\mathbf{c}):\mathrm{wt}(\mathbf{c})\ \mathrm{and}\ \mathbf{c}\in\mathcal{C}\}\},
\end{equation*}
where $\{\{\}\}$ is the multiset notation and $\frac{1}{q-1}S$ denotes the multiset resulting from dividing the multiplicity of each element in the multiset $S$ by $q-1$. 
If the pair $(\mathcal{P}(\mathcal{C}),\mathcal{B}_k(\mathcal{C}))$ is a $t$-$(n,k,\lambda)$ design, we say that the codewords of weight $k$ in $\mathcal{C}$ support a $t$-$(n,k,\lambda)$ design or the code $\mathcal{C}$ holds a $t$-$(n,k,\lambda)$ design. Assmus and Mattson established a sufficient condition for the pair $(\mathcal{P}(\mathcal{C}),\mathcal{B}_k(\mathcal{C}))$ to be a $t$-design \cite[Theorem 4.2]{Assmus1969NewDesign}. If $(\mathcal{P}(\mathcal{C}),\mathcal{B}_k(\mathcal{C}))$ is a $t$-$(n,k,\lambda)$ design satisfying the conditions of the Assmus-Mattson theorem, then it follows from \cite[Lemma 4.25]{Ding2018DesignsLinearCodes} that
\begin{equation}\label{sec1.2 equ2}
	b=\frac{A_k}{q-1}.
\end{equation}

\begin{theorem}[Assmus-Mattson Theorem]
	\label{sec1 thm1}
	Let $\mathcal{C}$ be a $[n,k,d]$ code over $\mathbb{F}_q$. Let $d^\perp$ denote the minimum distance of $\mathcal{C}^\perp$. Suppose that $w$ be the largest integer satisfying $w\le n$ and 
	\begin{equation*}
		w-\left\lfloor\frac{w+q-2}{q-1}\right\rfloor<d.
	\end{equation*}
	Define $w^\perp$ analogously using $d^\perp$. Denote the weight distribution of $\mathcal{C}$ and $\mathcal{C}^\perp$ by $(A_i)_{i=0}^n$ and $(A^\perp_i)_{i=0}^n$, respectively. Fix a positive integer with $t<d$. Let $s$ be the number of $i$ with $A_i^\perp \ne 0$ for $1\le i \le n-t$. Suppose $s\le d-t$. Then
	\begin{itemize}
		\item $(\mathcal{P}(\mathcal{C}),\mathcal{B}_i(\mathcal{C}))$ is a simple $t$-design provided $A_i \ne 0$ and $ d \le i \le w$, and 
		\item $(\mathcal{P}(\mathcal{C}^\perp),\mathcal{B}_i(\mathcal{C}^\perp))$ is a simple $t$-design provided $A^\perp_i \ne 0$ and $ d^\perp \le i \le \min\{n-t,w^\perp\}$.
	\end{itemize}
\end{theorem}

The Assmus-Mattson Theorem is very helpful for constructing $t$-designs from linear codes and has been widely used in \cite{Ding2018DesignsLinearCodes,Ding2020InfiniteFamiliesMDS,Tonchev1998CodesDesigns,Tonchev2007Codes,Heng2022ConstructionsMDSNear}. Apart from the Assmus-Mattson theorem, another sufficient condition for $(\mathcal{P},\mathcal{B}_k)$ to be a $t$-design is via the $t$-homogeneous or $t$-transitive automorphism group of the code $\mathcal{C}$ \cite[p. 308]{Huffman2003FundamentalErrorCodes}. Recently, Tang, Ding and Xiong \cite{Tang2020CodesDifferentiallyUniform} presented a generalization of the Assmus-Mattson Theorem. Here we outline the part of this generalization in preparation for the subsequent proofs. Notice that some of the $t$-designs from \cref{sec1 thm2} are trivial, and some may not be simple.
\begin{theorem}\cite[Theorem~5.3]{Tang2020CodesDifferentiallyUniform}
	\label{sec1 thm2}
	Let $\mathcal{C}$ be a $[n,k,d]$ code over $\mathbb{F}_q$. Let $d^\perp$ denote the minimum distance of $\mathcal{C}^\perp$. Let $r$ and $t$ be two positive integers with $t<\min\{d,d^\perp\}$. Let $R$ be an $r$-subset of $\{d,d+1,\dots,n-t\}$. Suppose that $(\mathcal{P}(\mathcal{C}),\mathcal{B}_{\ell}(\mathcal{C}))$ and $(\mathcal{P}(\mathcal{C}^\perp),\mathcal{B}_{\ell^\perp}(\mathcal{C}^\perp))$ are $t$-designs for $\ell\in\{d,d+1,\dots,n-t\}\setminus R$ and $0\le\ell^\perp\le r+t-1$. Then $(\mathcal{P}(\mathcal{C}),\mathcal{B}_i(\mathcal{C}))$ and $(\mathcal{P}(\mathcal{C}^\perp),\mathcal{B}_i(\mathcal{C}^\perp))$ are $t$-designs for any $t\le i\le n$.
\end{theorem}

\subsection{Some results of number theory}

The following are some results in elementary number theory, which will be useful for proving the theorems in \cref{sec3}.

\begin{lemma}
	\label{sec2.2 lem1}
	The congruence $ax \equiv b \pmod m$ has a solution if and only if $e=\gcd(a, m)$ divides $b$. Furthermore, if $e$ divides $b$, then the number of solutions of the congruence is $e$. Suppose that $x^\prime$ is one of the solutions. Then the $e$ solutions are $x^\prime, x^\prime+\frac{m}{e}, x^\prime+2\frac{m}{e},\dots,x^\prime+(e-1)\frac{m}{e}$.
\end{lemma}

\begin{lemma}\label{sec2.2 lem2}
	Let $p,i$ and $s$ be three positive integers. Suppose that $\gcd(i,s)=m$. Then
	\begin{equation*}
		\gcd(p^i+1,p^s+1)=\left\{\begin{array}{ll}
			p^m+1 &\mathrm{if}\ \frac{i}{m}\ \mathrm{and}\ \frac{s}{m}\ \mathrm{are\ odd},\\
			\begin{array}{ll}
				1&\mathrm{if}\ p\ \mathrm{is\ even,}\\
				2&\mathrm{if}\ p\ \mathrm{is\ odd,}\\
			\end{array} &\mathrm{otherwise}.\\
		\end{array}\right.
	\end{equation*}
\end{lemma}

\begin{lemma}\cite[Lemma 2.6]{Coulter1998Explicit}\cite[Lemma 2.1]{Coulter1999Evaluation}
	\label{sec2.2 lem3}
	Let $p,i$ and $s$ be three positive integers. Suppose that $\gcd(i,s)=m$. Then
	\begin{equation*}
		\gcd(p^i-1,p^s+1)=\left\{\begin{array}{ll}
			p^m+1 &\mathrm{if}\ \frac{i}{m}\ \mathrm{is\ even},\\
			\begin{array}{ll}
				1&\mathrm{if}\ p\ \mathrm{is\ even,}\\
				2&\mathrm{if}\ p\ \mathrm{is\ odd,}\\
			\end{array} &\mathrm{otherwise}.\\
		\end{array}\right.
	\end{equation*}
\end{lemma}

\subsection{The number of zeros of the polynomial $P_a(X)$}
We will need the following two lemmas about the zeros in $\mathbb{F}_{p^n}$ of the polynomial
\begin{equation*}
	P_a(X):=X^{p^k+1}+X+a,a\in\mathbb{F}_{p^n}^\ast.
\end{equation*}

\begin{lemma}\cite{Bluher2004Xq}
	\label{sec2.3 lem1}
	Let $e=\gcd(n,k)$ and $N_a$ denote the number of zeros in $\mathbb{F}_{p^n}$ of $P_a(X)$. Then $N_a$ takes either $0,1,2$ or $p^{e}+1$.
\end{lemma}

\begin{lemma}\cite[Lemma 17]{Kim2021COmpleteSolutionEquation}
	\label{sec2.3 lem2}
	Let $e=\gcd(n,k)$. Suppose that $P_a(X)=0$ has exactly $p^e+1$ solutions in $\mathbb{F}_{p^n}$. Let $x_0$ be a solution to $P_a(X)=0$. Then $\frac{x_{0}^{2}}{a}$ is a $(p^k-1)$-th power in $\mathbb{F}_{p^n}$. For $\delta\in\mathbb{F}_{p^n}$ with $\delta^{p^k-1}=\frac{x_{0}^{2}}{a}$,
	\begin{equation}\label{sec2.3 lem2 equ1}
		w^{p^k-1}-w+\frac{1}{\delta x_0}=0
	\end{equation}
	has exactly $p^e$ solutions in $\mathbb{F}_{p^n}$. Let $w_0\in\mathbb{F}_{p^n}$ be a solution to \cref{sec2.3 lem2 equ1}. Then, the $p^e+1$ solutions in $\mathbb{F}_{p^n}$ to $P_a(X)=0$ are $x_0,(w_0+\gamma)^{p^k-1}\cdot x_0$ where $\gamma$ runs over $\mathbb{F}_{p^e}$.
\end{lemma}

\section{Infinite families of MDS, NMDS and AMDS codes}\label{sec3}
Throughout this section and subsequent sections, let $q$ be the power of a prime $p$ and $U_l$ be the set of all $l$-th roots of unity in $\mathbb{F}_{q^2}$. It should be noted that for any $x\in U_{q+1}$ a straightforward yet crucial property is that $x^q=x^{-1}$. In this section, we show that the dual code of the BCH code $\mathcal{C}_{(q,q+1,3,h)}$ with $h=\frac{q-p^i}{2}$ or $\frac{p^i-1}{2}$ is a four-weight code. By determining the parameters of these four-weight codes, we obtain infinite families of MDS, NMDS and AMDS codes. Especially, we present a proof that resolves the conjecture posed in \cite[Conjecture 3.6]{Geng2022ClassAlmostMDS}.

\begin{theorem}\label{sec3 thm1}
	Let $q=p^s$ and $h=\frac{q-p^i}{2}$, where $s>1$ and $0<i<s$ are two integers. Suppose that $\gcd(i,s)=m$. Then the code $\mathcal{C}^\perp_{(q,q+1,3,h)}$ is a four-weight code with parameters $[q+1,4,q-p^m]$.
\end{theorem}
\begin{proof}
	Let $\alpha$ be a generator of $\mathbb{F}_{q^2}^{\ast}$. Then $\beta=\alpha^{q-1}$ is a primitive $(q+1)$-th root of unity in $\mathbb{F}_{q^2}$. Let $g_h(x)$ and $g_{h+1}(x)$ denote the minimal polynomials of $\beta^h$ and $\beta^{h+1}$ over $\mathbb{F}_q$, respectively. Note that $g_h(x)=(x-\beta^{\frac{q-p^i}{2}})(x-\beta^{\frac{q+p^i}{2}+1})$ and $g_{h+1}(x)=(x-\beta^{\frac{q-p^i}{2}+1})(x-\beta^{\frac{q+p^i}{2}})$. Since $0<i<s$, $\deg(g_h(x)g_{h+1}(x))=4$, and so the dimension of $\mathcal{C}^\perp_{(q,q+1,3,h)}$ is 4. Define
	\begin{equation*}
		\left.H=\left[\begin{array}{lllll}1&\beta^{h}&(\beta^{h})^2&\cdots&(\beta^{h})^q\\1&\beta^{h+1}&(\beta^{h+1})^2&\cdots&(\beta^{h+1})^q\end{array}\right.\right],
	\end{equation*}
	then $H$ is a parity-check matrix of $\mathcal{C}_{(q,q+1,3,h)}$. Then it follows from Delsarte’s theorem \cite[Theorem 2]{Delsarte1975SubfieldSubcodesModified} that the trace expression of $\mathcal{C}_{(q,q+1,3,h)}^{\perp}$ is given by
	\begin{equation*}
		\mathcal{C}_{(q,q+1,3,h)}^\perp=\{\mathbf{c}_{(a,b)}:a,b\in\mathbb{F}_{q^2}\},
	\end{equation*}
	where $\mathbf{c}_{(a,b)}=\left(\mathrm{Tr}_{q^{2}/q}\left(a\beta^{hi}+b\beta^{(h+1)i}\right)\right)_{i=0}^{q}.$ Let $u \in U_{q+1}$. Then 
	\begin{equation}
		\label{sec3 thm1 equ1}
		\begin{aligned}
			\operatorname{Tr}_{q^2/q}(au^h+bu^{h+1})&=au^h+bu^{h+1}+b^qu^{q-h}+a^qu^{q+1-h}\\&=u^{h}(a+bu+b^qu^{q-2h}+a^qu^{q-2h+1})\\&=u^{\frac{q-p^i}{2}}(a+bu+b^qu^{p^i}+a^qu^{p^i+1}).
		\end{aligned}
	\end{equation}
	Denote the number of solutions in $U_{q+1}$ of the equation
	\begin{equation}
		\label{sec3 thm1 equ2}
		a+bu+b^qu^{p^i}+a^qu^{p^i+1}=0
	\end{equation}
	by $N(a,b)$. Then $wt(\mathbf{c}_{(a,b)})=q+1-N(a,b)$.
	\begin{enumerate}[label = \bf Case \arabic*:, leftmargin=4.1em]
		\item If $a\ne0,b=0$, \cref{sec3 thm1 equ2} is equivalent to 
		\begin{equation}
			\label{sec3 thm1 equ3}
			1+a^{q-1}u^{p^i+1}=0.
		\end{equation}
		For fixed $a \in \mathbb{F}_{q^2}^\ast$, there exists $r$ such that $a=\alpha^r$. Suppose that $u=\alpha^{(q-1)t}$, then we have
		\begin{equation*}
			\begin{aligned}
				&\alpha^{(p^i+1)(q-1)t}=\alpha^{-(q-1)r},\ &\mathrm{if}\ p\ \mathrm{is\ even};\\
				&\alpha^{(p^i+1)(q-1)t}=\alpha^{\frac{q^2-1}{2}-(q-1)r},\ &\mathrm{if}\ p\ \mathrm{is\ odd}.
			\end{aligned}
		\end{equation*}
		Hence \cref{sec3 thm1 equ3} has solutions in $\mathbb{F}_{q^2}$ if and only if there exist $t\in\mathbb{N}$ such that
		\begin{equation}
			\begin{aligned}
				&\left(p^i+1\right)t \equiv -r \pmod {q+1},\ &\mathrm{if}\ p\ \mathrm{is\ even};\\
				&\left(p^i+1\right)t \equiv \frac{q+1}{2}-r \pmod {q+1},\ &\mathrm{if}\ p\ \mathrm{is\ odd}.
			\end{aligned}
		\end{equation}
		According to \cref{sec2.2 lem1,sec2.2 lem2}, we have
		\begin{equation*}
			\begin{array}{ll}
				N(a,b)=\left\{\begin{array}{ll}
					p^m+1, &\mathrm{if}\ p^m+1\mid -r,\ \mathrm{and}\ \frac{i}{m}\ \mathrm{and}\ \frac{s}{m}\ \mathrm{are\ odd},\\
					1, &\mathrm{otherwise},\end{array}\right. &\mathrm{if}\ p\ \mathrm{is\ even;}\\
				N(a,b)=\left\{\begin{array}{ll}
					p^m+1, &\mathrm{if}\ p^m+1\mid \frac{q+1}{2}-r,\ \mathrm{and}\ \frac{s}{m}\ \mathrm{and}\ \frac{i}{m}\ \mathrm{are\ odd},\\
					2, &\mathrm{if}\ 2\mid \frac{q+1}{2}-r,\ \mathrm{and}\ \frac{i}{m}\ \mathrm{or}\ \frac{s}{m}\ \mathrm{is\ even},\\
					0,&\mathrm{otherwise,}\end{array}\right. &\mathrm{if}\ p\ \mathrm{is\ odd.}\\
			\end{array}
		\end{equation*}
		
		\item If $a=0,b\ne0$, \cref{sec3 thm1 equ2} can be rewritten as
		\begin{equation}
			\label{sec3 thm1 equ5}
			1+b^{q-1}u^{p^i-1}=0.
		\end{equation}
		For fixed $b \in \mathbb{F}_{q^2}^\ast$, there exists $r$ such that $b=\alpha^r$. Suppose that $u=\alpha^{(q-1)t}$. Thus we have
		\begin{equation*}
			\begin{aligned}
				&\alpha^{(p^i-1)(q-1)t}=\alpha^{-(q-1)r},\ &\mathrm{if}\ p\ \mathrm{is\ even};\\
				&\alpha^{(p^i-1)(q-1)t}=\alpha^{\frac{q^2-1}{2}-(q-1)r},\ &\mathrm{if}\ p\ \mathrm{is\ odd}.
			\end{aligned}
		\end{equation*}
		Hence \cref{sec3 thm1 equ5} has solutions in $\mathbb{F}_{q^2}$ if and only if there exist $t\in\mathbb{N}$ such that
		\begin{equation}
			\label{sec3 thm1 equ6}
			\begin{aligned}
				&\left(p^i-1\right)t \equiv -r \pmod {q+1},\ &\mathrm{if}\ p\ \mathrm{is\ even};\\
				&\left(p^i-1\right)t \equiv \frac{q+1}{2}-r \pmod {q+1},\ &\mathrm{if}\ p\ \mathrm{is\ odd}.
			\end{aligned}
		\end{equation}
		By \cref{sec2.2 lem1,sec2.2 lem3}, we get 
		\begin{equation*}
			\begin{array}{ll}
				N(a,b)=\left\{\begin{array}{ll}
					p^m+1, &\mathrm{if}\ p^m+1\mid -r,\ \mathrm{and}\ \frac{i}{m}\ \mathrm{is\ even,}\\
					1, &\mathrm{otherwise},\end{array}\right. &\mathrm{if}\ p\ \mathrm{is\ even;}\\
				N(a,b)=\left\{\begin{array}{ll}
					p^m+1, &\mathrm{if}\ p^m+1\mid \frac{q+1}{2}-r,\ \mathrm{and}\ \frac{i}{m}\ \mathrm{is\ even,}\\
					2, &\mathrm{if}\ 2\mid \frac{q+1}{2}-r,\ \mathrm{and}\ \frac{i}{m}\ \mathrm{is\ odd},\\
					0,&\mathrm{otherwise,}\end{array}\right. &\mathrm{if}\ p\ \mathrm{is\ odd.}\\
			\end{array}
		\end{equation*}	
		
		\item If $a\ne0,b\ne0$, \cref{sec3 thm1 equ2} can be rewritten as
		\begin{equation*}
			u^{p^i+1}+xu^{p^i}+yu+z=0,
		\end{equation*}
		where $x=\frac{b^q}{a^q}$, $y=\frac{b}{a^q}$ and $z=a^{1-q}$. 
		\begin{enumerate}[label=\bf Case 3.\arabic*:]
			\item If $y= x^{p^i}$ and $z=xy$, then
			\begin{equation*}
				u^{p^i+1}+xu^{p^i}+yu+z=\left(u+x\right)^{p^i+1}=0.
			\end{equation*}
			Since $z=xy$, we have $a^{q+1}-b^{q+1}=0$. Then $x \in U_{q+1}$, so is $-x$. Hence, $N(a,b)=1$.
			
			\item If $y\ne x^{p^i}$ and $z=xy$, we have
			\begin{equation*}
				u^{p^i+1}+xu^{p^i}+yu+z=\left(u^{p^i}+y\right)\left(u+x\right)=0.
			\end{equation*}
			Since $-x\in U_{q+1}$, $N(a,b)\in \{1,2\}$.
			
			\item If $y= x^{p^i}$ and $z\ne xy$, we have 
			\begin{equation*}
				u^{p^i+1}+xu^{p^i}+yu+z=\left(u+x\right)^{p^i+1}+\left(z-xy\right)=0.
			\end{equation*}
			Substituting $u=v-x$ yields
			\begin{equation}\label{sec3 thm1 equ7}
				v^{p^i+1}=xy-z.
			\end{equation}
			For fixed $a,b\in\mathbb{F}_{q^2}^\ast$, there exists $r$ such that $xy-z=\alpha^r$. Suppose that $v=\alpha^t$, then 
			\begin{equation*}
				\alpha^{\left(p^i+1\right)t}=\alpha^{r}.
			\end{equation*}
			This equation has solutions in $\mathbb{F}_{q^2}$ if and only if
			\begin{equation*}
				\left(p^i+1\right)t \equiv r \pmod {q^2-1}.
			\end{equation*}
			Let $e=\gcd(p^i+1,q^2-1)$. It follows from \cref{sec2.2 lem1} that \cref{sec3 thm1 equ7} has a solution if and only if $e \mid r$. Now assume that $e \mid r$. Then \cref{sec3 thm1 equ7} has $e$ solution(s). Suppose that $v_0$ is one of the solutions of \cref{sec3 thm1 equ7}. Then $V=\{v_0\cdot\alpha^{\frac{q^2-1}{e}k} : k=0,1,\dots,e-1\}$ is the set of all solutions in $\mathbb{F}_{q^2}$ of \cref{sec3 thm1 equ7}. Since $u$ and $v$ are in one-to-one correspondence, $N(a,b)$ is equal to the number of $k \in \{0,1,\dots,e-1\}$ such that
			\begin{equation}\label{sec3 thm1 equ8}
				\left(v_0\cdot\alpha^{\frac{q^2-1}{e}k}-x\right)^{q+1}=1.
			\end{equation}
			Expansion on the left side of this equation gives
			\begin{equation}\label{sec3 thm1 equ9}
				v_0^{q+1}\cdot\alpha^{\frac{q^2-1}{e}k(q+1)}-xv_0^q\cdot\alpha^{\frac{q^2-1}{e}kq}-x^qv_0\cdot\alpha^{\frac{q^2-1}{e}k}+x^{q+1}-1=0.
			\end{equation}
			If $e \mid q+1$, then $\alpha^{\frac{q^2-1}{e}k(q+1)}=1$ and $\alpha^{\frac{q^2-1}{e}kq}=\alpha^{-\frac{q^2-1}{e}k}$. Hence, we have
			\begin{equation*}
				v_0^{q+1}-xv_0^q\cdot\alpha^{-\frac{q^2-1}{e}k}-x^qv_0\cdot\alpha^{\frac{q^2-1}{e}k}+x^{q+1}-1=0.
			\end{equation*}
			Multiplying both sides of this equation by $\alpha^{\frac{q^2-1}{e}k}$, we get
			\begin{equation}\label{sec3 thm1 equ10}
				-x^qv_0\cdot\left(\alpha^{\frac{q^2-1}{e}k}\right)^2+\left(v_0^{q+1}+x^{q+1}-1\right)\cdot\alpha^{\frac{q^2-1}{e}k}-xv_0^q=0.
			\end{equation}
			If $e \nmid q+1$, then it follows from \cref{sec2.2 lem3} that $e=p^{\gcd(i,2s)}+1$ since $1 \mid q+1$ if $p$ is even and $2 \mid q+1$ if $p$ is odd. Recall that $m=\gcd(i,s)$. If $\gcd(i,2s)=m$, then $p^m+1 \nmid q+1$, which indicates that $\frac{s}{m}$ is even by \cref{sec2.2 lem2}. If $\gcd(i,2s)=2m$, then $\frac{s}{m}$ is even as $\frac{2s}{\gcd(i,2s)}$ is even according to \cref{sec2.2 lem3}. Thus, $q+1 \equiv (p^m)^{\frac{s}{m}}+1 \equiv (-1)^{\frac{s}{m}}+1 \equiv 2 \bmod e$. Hence, we have $\alpha^{\frac{q^2-1}{e}k(q+1)}=\alpha^{2\frac{q^2-1} {e}k}$ and $\alpha^{\frac{q^2-1}{e}kq}=\alpha^{\frac{q^2-1}{e}k}$. Then \cref{sec3 thm1 equ9} leads to
			\begin{equation}\label{sec3 thm1 equ11}
				v_0^{q+1}\cdot\left(\alpha^{\frac{q^2-1}{e}k}\right)^2-\left(xv_0^q+x^qv_0\right)\cdot\alpha^{\frac{q^2-1}{e}k}+x^{q+1}-1=0.
			\end{equation}
			Combining \cref{sec3 thm1 equ10,sec3 thm1 equ11}, we have $N(a,b)\le2$, i.e., $N(a,b) \in \{0,1,2\}$.     
			
			\item If $y\ne x^{p^i}$ and $z\ne xy$, substituting $\left(y-x^{p^i}\right)^{\frac{1}{p^i}}v-x$ for $u$ leads to
			\begin{equation*}
				u^{p^i+1}+xu^{p^i}+yu+z
				=\left(y- x^{p^i}\right)^{\frac{p^i+1}{p^i}}\left(v^{p^i+1}+v+\frac{z-xy}{\left(y- x^{p^i}\right)^{\frac{p^i+1}{p^i}}}\right)=0.
			\end{equation*}
			Let $e=\gcd(i, 2s)$ and $D=\frac{z-xy}{\left(y- x^{p^i}\right)^{\frac{p^i+1}{p^i}}}$. By \cref{sec2.3 lem1}, the number of solutions of 
			\begin{equation}\label{sec3 thm1 equ12}
				v^{p^i+1}+v+D=0
			\end{equation} 
			takes either $0,1,2$ or $p^{e}+1$. Given the values of $N(a,b)$ determined before, we need to consider the only case where \cref{sec3 thm1 equ12} has $p^{e}+1$ solutions. 
			
			Suppose that \cref{sec3 thm1 equ12} has $p^{e}+1$ solutions in $\mathbb{F}_{q^2}$. Let $v_0$ be an arbitrary solution to \cref{sec3 thm1 equ12}. According to \cref{sec2.3 lem2},  $\frac{v_{0}^{2}}{D}$ is a $(p^i-1)$-th power in $\mathbb{F}_{q^2}$. For $\delta\in\mathbb{F}_{q^2}$ with $\delta^{p^i-1}=\frac{v_{0}^{2}}{D}$,
			\begin{equation}\label{sec3 thm1 equ13}
				w^{p^i-1}-w+\frac{1}{\delta v_0}=0
			\end{equation}
			has exactly $p^e$ solutions in $\mathbb{F}_{q^2}$. Let $w_0\in\mathbb{F}_{q^2}$ be a solution to \cref{sec3 thm1 equ13}. Then, the $p^e+1$ solutions in $\mathbb{F}_{q^2}$ to \cref{sec3 thm1 equ12} are $v_0,(w_0+\gamma)^{p^i-1}\cdot v_0$ where $\gamma$ runs over $\mathbb{F}_{p^e}$. Denote the set of these $p^e+1$ solutions by $V$. Then $N(a,b)$ is equal to the number of $v \in V$ such that 
			\begin{equation}
				\label{sec3 thm1 equ14}
				(Av-B)^{q+1}-1=A^{q+1}v^{q+1}-BA^qv^q-B^qAv+B^{q+1}-1=0,
			\end{equation}
			where $A=\left(y-x^{p^i}\right)^{\frac{1}{p^i}}$ and $B=x$. Now consider the case of $N(a,b)>0$. Without loss of generality, assume that $(Av_0-B)^{q+1}-1=0$ where $v_0 \in V$. For fixed $v_0,w_0$ and $\delta$, we then check the number of $\gamma\in\mathbb{F}_{p^e}$ satisfying
			\begin{equation}\label{sec3 thm1 equ15}
				\begin{aligned}
					&A^{q+1}\left[(w_0+\gamma)^{p^i-1}\cdot v_0\right]^{q+1}-BA^q\left[(w_0+\gamma)^{p^i-1}\cdot v_0\right]^q\\
					-&B^qA\left[(w_0+\gamma)^{p^i-1}\cdot v_0\right]+B^{q+1}-1=0.
				\end{aligned}
			\end{equation}
			Note that $w_0+\gamma\ne 0$. Multiplying both sides of \cref{sec3 thm1 equ15} by $(w_0+\gamma)^{q+1}$ yields
			\begin{equation*}
				\begin{aligned}
					&A^{q+1}v_0^{q+1}\left[(w_0+\gamma)^{p^i}\right]^{q+1}-BA^qv_0^{q}(w_0+\gamma)\left[(w_0+\gamma)^{p^i}\right]^q\\
					&-B^qAv_0(w_0+\gamma)^{q}(w_0+\gamma)^{p^i}+\left(B^{q+1}-1\right)(w_0+\gamma)^{q+1}=0
				\end{aligned}
			\end{equation*}
			Since $\gamma\in\mathbb{F}_{p^e}$ and $e | i$, $\gamma^{p^e}=\gamma^{p^i}=\gamma$. Then $(w_0+\gamma)^{p^i}=w_0^{p^i}+\gamma$, which leads to
			\begin{equation*}
				D_1\gamma^{q+1}+D_2\gamma^q+D_3\gamma+D_4=0,
			\end{equation*}
			where
			\begin{equation*}
				\begin{aligned}
					&D_1=(Av_0-B)^{q+1}-1=0, D_2=-\frac{A}{\delta}(Av_0-B)^q,\\
					&D_3=-\frac{A^q}{\delta^q}(Av_0-B),D_4=w_0^{q+1}\left[(Av_0w_0^{p^i-1}-B)^{q+1}-1\right].
				\end{aligned}
			\end{equation*}
			It is clear that $D_3 \ne 0$ and $D_2=D_3^q$. Thus, \cref{sec3 thm1 equ15} is equivalent to
			\begin{equation}
				\label{sec3 thm1 equ16}
				(D_3\gamma)^q+D_3\gamma+D_4=0.
			\end{equation}
			which is a linearized or affine affine polynomial over $\mathbb{F}_{q^2}$. Note that $\{x\in \mathbb{F}_{q^2} \mid x+x^q=0\}=(\alpha-\alpha^q)\mathbb{F}_q$. Let $S=D_3^{-1}(\alpha-\alpha^q)\mathbb{F}_q \cap \mathbb{F}_{p^e}$. 
			
			\begin{enumerate}[label = \bf Case]
				\item \textbf{: $D_4 = 0$.}\\
				Now $S$ is the set of all solutions in $\mathbb{F}_{p^e}$ of \cref{sec3 thm1 equ16}. One can deduce that $S$ is a linear space over $\mathbb{F}_{p^m}$. Recall that $m=\gcd(i,s)$ and $e=\gcd(i,2s)$. Since $\mathbb{F}_{p^e}$ is a $\frac{e}{m}$-dimensional linear space over $\mathbb{F}_{p^m}$, and $S$ is a linear subspace of $\mathbb{F}_{p^e}$, the dimension of $S$ is either 0, 1, or 2. Assume on the contrary that $S$ is a 2-dimensional linear space over $\mathbb{F}_{p^m}$, then we have $e=2m$ and $S=\mathbb{F}_{p^{2m}}$. If $D_3^{-1}(\alpha-\alpha^q) \in \mathbb{F}_{q}^\ast$, then $S=\mathbb{F}_q \cap \mathbb{F}_{p^e}=\mathbb{F}_{p^m}$ as $\gcd(s,e)=m$, a contradiction. If $D_3^{-1}(\alpha-\alpha^q) \notin \mathbb{F}_{q}^\ast$, then $S\subseteq \mathbb{F}_{p^{2m}}\setminus\mathbb{F}_{p^m}\cup\{0\}$, and so $|S| \le p^{2m}-p^m+1$, which contradicts to $|S| = p^{2m}$. Hence, the dimension of $S$ over $\mathbb{F}_{p^m}$ is either 0 or 1, and so $N(a,b)=2$ or $p^m+1$.
				\item \textbf{: $D_4 \ne 0$.}\\
				If \cref{sec3 thm1 equ16} does not have a solution in $\mathbb{F}_{p^e}$, then $N(a,b)=1$. If there exists $\gamma_0 \in \mathbb{F}_{p^e}$ such that \cref{sec3 thm1 equ16} holds, then $\gamma_0+S$ is the set of all solutions in $\mathbb{F}_{p^e}$ of \cref{sec3 thm1 equ16}, which is an affine space over $\mathbb{F}_{p^m}$. It is clear that the dimension of this affine space is either 0 or 1, which implies that $N(a,b)=2$ or $p^m+1$.
			\end{enumerate}
			Combining the above cases with the case where \cref{sec3 thm1 equ12} has either 0, 1 or 2 solutions in $\mathbb{F}_{q^2}$, we get $N(a,b)$ takes either 0,1,2 or $p^m+1$.
		\end{enumerate}
	\end{enumerate}
	In summary, $N(a,b)$ takes either 0,1,2 or $p^m+1$. By the Singleton bound, there must exists a codeword $\in \mathcal{C}^\perp_{(q,q+1,3,h)}$ of weight $q+1-(p^m+1)$. Therefore, the code $\mathcal{C}^\perp_{(q,q+1,3,h)}$ is a four-weight code with parameters $[q+1,4,q-p^m]$.
\end{proof}

According to \cref{sec3 thm1}, we obtain an infinite families of MDS code when $p^m=2$.
\begin{corollary}\label{sec3 cor1}
	Let $q=2^s$ and $h=\frac{q-2^i}{2}$, where $s>1$ and $0<i<s$ are two integers. If $\gcd(i,s)=1$, then the BCH code $\mathcal{C}_{(q,q+1,3,h)}$ is an MDS code with parameters $[q+1,q-3,5]$, and its dual code $\mathcal{C}^\perp_{(q,q+1,3,h)}$ has parameters $[q+1,4,q-2]$.
\end{corollary}
\begin{proof}
	Since $\gcd(i,s)=1$, it follows from \cref{sec3 thm1} that $\mathcal{C}^\perp_{(q,q+1,3,h)}$ is an MDS code with parameters $[q+1,4,q-2]$. Thus $\mathcal{C}_{(q,q+1,3,h)}$ is also an MDS code with parameters $[q+1,q-3,5]$.
\end{proof}

Reed-Solomon codes over $\mathbb{F}_q$ are MDS codes with length $q-1$ \cite[pp. 173]{Huffman2003FundamentalErrorCodes}. Hence, the codes $\mathcal{C}_{(q,q+1,3,h)}$ and $\mathcal{C}^\perp_{(q,q+1,3,h)}$ in \cref{sec3 cor1} are not Reed-Solomon codes.

When $p$ is an odd prime and $h=\frac{p^i-1}{2}$, we present that the code $\mathcal{C}^\perp_{(q,q+1,3,h)}$ is also a four-weight code.

\begin{theorem}\label{sec3 thm4}
	Let $q=p^s$ and $h=\frac{p^i-1}{2}$, where $p$ is odd, and $s>1$ and $0<i<s$ are two integers. Suppose that $\gcd(i,s)=m$. Then the code $\mathcal{C}^\perp_{(q,q+1,3,h)}$ is a four-weight code with parameters $[q+1,4,q-p^m]$.
\end{theorem}
\begin{proof}
	Let $\alpha$ be a generator of $\mathbb{F}_{q^2}^{\ast}$. Then $\beta=\alpha^{q-1}$ is a primitive $(q+1)$-th root of unity in $\mathbb{F}_{q^2}$. Let $g_h(x)$ and $g_{h+1}(x)$ denote the minimal polynomials of $\beta^h$ and $\beta^{h+1}$ over $\mathbb{F}_q$, respectively. Note that $g_h(x)=(x-\beta^{\frac{p^i-1}{2}})(x-\beta^{q-\frac{p^i+1}{2}})$ and $g_{h+1}(x)=(x-\beta^{\frac{p^i+1}{2}})(x-\beta^{q-\frac{p^i-1}{2}})$. Since $0<i<s$, $\deg(g_h(x)g_{h+1}(x))=4$, and so the dimension of $\mathcal{C}^\perp_{(q,q+1,3,h)}$ is 4. Define
	\begin{equation*}
		\left.H=\left[\begin{array}{lllll}1&\beta^{h}&(\beta^{h})^2&\cdots&(\beta^{h})^q\\1&\beta^{h+1}&(\beta^{h+1})^2&\cdots&(\beta^{h+1})^q\end{array}\right.\right],
	\end{equation*}
	then $H$ is a parity-check matrix of $\mathcal{C}_{(q,q+1,3,h)}$. Then it follows from Delsarte’s theorem that the trace expression of $\mathcal{C}_{(q,q+1,3,h)}^{\perp}$ is given by
	\begin{equation*}
		\mathcal{C}_{(q,q+1,3,h)}^\perp=\{\mathbf{c}_{(a,b)}:a,b\in\mathbb{F}_{q^2}\},
	\end{equation*}
	where $\mathbf{c}_{(a,b)}=\left(\mathrm{Tr}_{q^{2}/q}\left(a\beta^{hi}+b\beta^{(h+1)i}\right)\right)_{i=0}^{q}.$ Let $u \in U_{q+1}$. Then 
	\begin{equation}
		\label{sec3 thm4 equ1}
		\begin{aligned}
			\operatorname{Tr}_{q^2/q}(au^h+bu^{h+1})&=au^h+bu^{h+1}+a^qu^{-h}+b^qu^{-h-1}\\&=u^{-h-1}(b^q+a^qu+au^{2h+1}+bu^{2h+2})\\&=u^{-\frac{p^i-1}{2}-1}(b^q+a^qu+au^{p^i}+bu^{p^i+1}).
		\end{aligned}
	\end{equation}
	Denote the number of solutions in $U_{q+1}$ of the equation
	\begin{equation}
		\label{sec3 thm4 equ2}
		b^q+a^qu+au^{p^i}+bu^{p^i+1}=0
	\end{equation}
	by $N(a,b)$. Then $wt(\mathbf{c}_{(a,b)})=q+1-N(a,b)$. We omit the remaining part, since it is similar to the proof of \cref{sec3 thm1}. 
\end{proof}

The following theorem will be helpful to determine the minimum distance of the BCH code $\mathcal{C}_{(q,q+1,3,h)}$.
\begin{theorem}\cite[Theorem 2]{Xu2023InfiniteFamiliesMDS}
	\label{sec3 thm5}
	The minimum distance $d$ of the BCH code $\mathcal{C}_{(q,q+1,3,h)}$ is equal to 3 if and only if $\gcd(2h+1,q+1)>1$. Equivalently, $d \ge 4$ if and only if $\gcd(2h+1,q+1)=1$.
\end{theorem}

Combining \cref{sec3 thm1,sec3 thm4,sec3 thm5}, infinite families of AMDS codes for any $p$ are given as follows. 
\begin{theorem}
	\label{sec3 thm6}
	Let $q=p^s$ where $s>1$ is an integer. Let $h=\frac{q-p^i}{2}$ for any $p$ or $h=\frac{p^i-1}{2}$ for $p$ odd, where $0<i<s$ is an integer. Suppose that $\gcd(i,s)=m$. If $p^m\ge3$, then the BCH code $\mathcal{C}_{(q,q+1,3,h)}$ is an AMDS code with parameters $[q+1,q-3,4]$, and its dual code $\mathcal{C}^\perp_{(q,q+1,3,h)}$ has parameters $[q+1,4,q-p^m]$.
\end{theorem}
\begin{proof}
	Let $d$ be the minimum distance of the BCH code $\mathcal{C}_{(q,q+1,3,h)}$. Since $\gcd(2h+1,q+1)=\gcd(p^i,q+1)=1$, $d\ge4$ by \cref{sec3 thm5}. Given that $\gcd(i,s)=m$, it follows from \cref{sec3 thm1,sec3 thm4} that $\mathcal{C}^\perp_{(q,q+1,3,h)}$ is has parameters $[q+1,4,q-p^m]$. Note that $p^m\ge 3$. According to the Singleton bound, $\mathcal{C}^\perp_{(q,q+1,3,h)}$ is not MDS, and neither is $\mathcal{C}_{(q,q+1,3,h)}$. Thus,  $\mathcal{C}_{(q,q+1,3,h)}$ an AMDS code with parameters $[q+1,q-3,4]$.
\end{proof}

\begin{remark}
	The weight enumerator of the code $\mathcal{C}^\perp_{(q,q+1,3,h)}$ in \cref{sec3 thm6} will be established through the MacWilliams equations and the parameters of the $3$-design in \cref{sec4 thm1}.
\end{remark}

When $p=3$, a straightforward result from \cref{sec3 thm6} regarding infinite families of NMDS codes is shown as follows.
\begin{corollary}
	\label{sec3 cor2}
	Let $q=3^s$, and $h=\frac{q-3^i}{2}$ or $\frac{3^i-1}{2}$, where $s>1$ and $0<i<s$ are two integers. If $\gcd(i,s)=1$, then the BCH code $\mathcal{C}_{(q,q+1,3,h)}$ is an NMDS code with parameters $[q+1,q-3,4]$, and its dual code $\mathcal{C}^\perp_{(q,q+1,3,h)}$ has parameters $[q+1,4,q-3]$.
\end{corollary}

Geng et al. conjecture that $\mathcal{C}^\perp_{(3^s,3^s+1,3,4)}$ with $s$ being odd is an AMDS code with parameters $[q+1,4,q-3]$ in \cite[Conjecture 3.6]{Geng2022ClassAlmostMDS}. The following corollary shows that this conjecture is indeed true. Actually, $\mathcal{C}^\perp_{(3^s,3^s+1,3,4)}$ is a specific case of \cref{sec3 cor2}. 
\begin{corollary}
	\label{sec3 cor3}
	Let $q=3^s$ with $s$ being odd. Then the BCH code $\mathcal{C}_{(q,q+1,3,4)}$ is an NMDS code with parameters $[q+1,q-3,4]$, and its dual code $\mathcal{C}^\perp_{(q,q+1,3,4)}$ has parameters $[q+1,4,q-3]$.
\end{corollary}
\begin{proof}
	Note that $4=\frac{3^2-1}{2}$. Since $s$ is odd, $\gcd(2,s)=1$. Then the proof is completed by \cref{sec3 cor2}.
\end{proof}

\section{Infinite families NMDS and AMDS codes holding 3-designs}\label{sec4}
In this section, we prove that the minimum weight codewords of the AMDS codes $\mathcal{C}_{(q,q+1,3,h)}$ and their duals in \cref{sec3 thm6} support $3$-designs for any prime power $q$, respectively. Moreover, we show that the codewords of weight 5 in these NMDS codes in \cref{sec3 cor2} also support a $3$-design.

\begin{theorem}
	\label{sec4 thm1}
	Let $q=p^s$ where $s>1$ is an integer. Let $h=\frac{q-p^i}{2}$ for any $p$ or $h=\frac{p^i-1}{2}$ for $p$ odd, where $0<i<s$ is an integer. Suppose that $\gcd(i,s)=m$. If $p^m\ge3$, the minimum weight codewords in $\mathcal{C}_{(q,q+1,3,h)}$ support a $3$-$(q+1,4,p^m-2)$ simple design, and the minimum weight codewords in $\mathcal{C}^\perp_{(q,q+1,3,h)}$ support a $3$-$(q+1,q-p^m,\lambda)$ simple design with
	\begin{equation*}
		\lambda=\frac{(q-p^m)(q-p^m-1)(q-p^m-2)}{(p^{2m}-1)p^{m}}.
	\end{equation*}
	Moreover, the weight enumerator of $\mathcal{C}^\perp_{(q,q+1,3,h)}$ is
	\begin{equation*}
		\begin{aligned}
			1+\frac{(q-1)^2q(q+1)}{(p^{2m}-1)p^{m}}z^{q-p^m}+\frac{(q^2-1)q((q+1)(p^m-1)-(q-1))}{2(p^m-1)}z^{q-1}+\\\frac{(q^2-1)(q^2-q+p^m)}{p^m}z^q+\frac{p^m(q-1)^2q(q+1)}{2(p^m+1)}z^{q+1}.
		\end{aligned}
	\end{equation*}
\end{theorem}
\begin{proof}
	According to \cref{sec3 thm6}, the parameters of $\mathcal{C}_{(q,q+1,3,h)}$ and $\mathcal{C}^\perp_{(q,q+1,3,h)}$ are $[q+1,q-3,4]$ and $[q+1,4,q-p^m]$, respectively. 
	
	Let $h=\frac{q-p^i}{2}$. Recall from the proof of \cref{sec3 thm1} that $\beta$ is a primitive $(q+1)$-th root of unity in $\mathbb{F}_{q^2}$, and 
	\begin{equation*}
		\left.H_1=\left[\begin{array}{lllll}1&\beta^{h}&(\beta^{h})^2&\cdots&(\beta^{h})^q\\1&\beta^{h+1}&(\beta^{h+1})^2&\cdots&(\beta^{h+1})^q\\1&\beta^{q-h}&(\beta^{q-h})^2&\cdots&(\beta^{q-h})^q\\1&\beta^{q-h+1}&(\beta^{q-h+1})^2&\cdots&(\beta^{q-h+1})^q\end{array}\right.\right]
	\end{equation*}
	is a parity-check matrix of $\mathcal{C}_{(q,q+1,3,h)}$. Let $x,y,z,w$ be four pairwise distinct elements in $U_{q+1}$. Define
	\begin{equation*}
		M_1=\left[\begin{array}{llll}x^{h}&y^{h}&z^{h}&w^{h}\\x^{h+1}&y^{h+1}&z^{h+1}&w^{h+1}\\x^{q-h}&y^{q-h}&z^{q-h}&w^{q-h}\\x^{q-h+1}&y^{q-h+1}&z^{q-h+1}&w^{q-h+1}\end{array}\right].
	\end{equation*}
	Then we have
	\begin{equation*}
		|M_1|=x^{\frac{q-p^i}{2}}y^{\frac{q-p^i}{2}}z^{\frac{q-p^i}{2}}w^{\frac{q-p^i}{2}}\left|\begin{array}{llll}1&1&1&1\\x&y&z&w\\x^{p^i}&y^{p^i}&z^{p^i}&w^{p^i}\\x^{p^i+1}&y^{p^i+1}&z^{p^i+1}&w^{p^i+1}\end{array}\right|.
	\end{equation*}
	Define
	\begin{equation}\label{sec4 thm1 equ1}
		M(x,y,z,w)=\left[\begin{array}{llll}1&1&1&1\\x&y&z&w\\x^{p^i}&y^{p^i}&z^{p^i}&w^{p^i}\\x^{p^i+1}&y^{p^i+1}&z^{p^i+1}&w^{p^i+1}\end{array}\right].
	\end{equation}
	It is clear that $|M_1|=0$ if and only if $|M(x,y,z,w)|=0$.	We next show that the total number of different choices of $w$ such that $|M_1|=0$ is equal to $p^m-2$ for any three different elements $x,y,z$. For fixed $x,y,z \in U_{q+1}$, let $$f(w)=|M(x,y,z,w)|=D_{p^i+1}w^{p^i+1}+D_{p^i}w^{p^i}+D_{1}w+D_{0},$$
	where $D_j$ is the cofactor of $w^j$ for $j \in \{0,1,p^i,p^i+1\}$. Since the minimum distance of $\mathcal{C}_{(q,q+1,3,h)}$ is 4, $D_j$ is not equal 0 for any $j$. It then follows from Case 3 in the proof of \cref{sec3 thm1} that $f(w)$ has either 0,1,2 or $p^{m} + 1$ zeros in $U_{q+1}$. Note that $f(x)=f(y)=f(z)=0$. Thus, $x,y$ and $z$ are three distinct zeros in $U_{q+1}$ of $f(w)$. Since $p^m\ge 3$, $f(w)$ has $p^{m} + 1$ zeros in $U_{q+1}$, which implies that there exist $(p^m+1-3)$ $w$’$s \in U_{q+1}\setminus\{x,y,z\}$ such that $f(w)=0$.
	
	Note that the minimum distance of $\mathcal{C}_{(q,q+1,3,h)}$ is 4. Let $\mathbf{c}=(c_0,c_1,\dots,c_q)$ be a codeword of weight 4 in $\mathcal{C}_{(q,q+1,3,h)}$ with nonzero coordinates $\{i_{1},i_{2},i_{3},i_{4}\}$. Suppose that
	\begin{equation*}
		x=\beta^{i_1}, y=\beta^{i_2}, z=\beta^{i_3}, w=\beta^{i_4}.
	\end{equation*}
	Then
	\begin{equation*}
		\left[\begin{array}{llll}x^{h}&y^{h}&z^{h}&w^{h}\\x^{h+1}&y^{h+1}&z^{h+1}&w^{h+1}\\x^{q-h}&y^{q-h}&z^{q-h}&w^{q-h}\\x^{q-h+1}&y^{q-h+1}&z^{q-h+1}&w^{q-h+1}\end{array}\right]\begin{bmatrix}c_{i_1}\\c_{i_2}\\c_{i_3}\\c_{i_4}\end{bmatrix}=\mathbf{0}.
	\end{equation*}
	Since the rank of $M_1$ is 3 for these $x,y,z,w$, the set $\{a\mathbf{c}:a\in \mathbb{F}_{q}^\ast\}$ consists of all such codewords of weight 4 in $\mathcal{C}_{(q,q+1,3,h)}$ with nonzero coordinates $\{i_{1},i_{2},i_{3},i_{4}\}$. Conversely, every codeword of weight 4 in $\mathcal{C}_{(q,q+1,3,h)}$ with nonzero coordinates $\{i_{1},i_{2},i_{3},i_{4}\}$ must correspond to the set $\{x,y,z,w\}$. Since the number of different choices of $w$ such that $M_1=0$ for fixed $x,y,z$ is $p^m-2$, we then deduce that the codewords of weight 4 in $\mathcal{C}_{(q,q+1,3,h)}$	support a $3$-$(q+1,4,p^m-2)$ simple design. By \cref{sec1.2 equ1}, the number of codewords of weight 4 in $\mathcal{C}_{(q,q+1,3,h)}$ is given by
	\begin{equation*}
		A_4=(q-1)(p^m-2)\frac{\binom{q+1}{3}}{\binom{4}{3}}=\frac{(p^m-2)(q-1)^2q(q+1)}{24}.
	\end{equation*}
	
	Let $(A_i)_{i=0}^n$ and $(A^\perp_i)_{i=0}^n$ denote the weight distributions of $\mathcal{C}^\perp_{(q,q+1,3,h)}$ and $\mathcal{C}^\perp_{(q,q+1,3,h)}$ separately. By setting $j=q-3$ in \cref{sec1.1 equ1}, we get
	\begin{equation*}
		A^\perp_{q-p^m}=A_4\bigg/\binom{p^m+1}{4}=\frac{(q-1)^2q(q+1)}{(p^{2m}-1)p^{m}}.
	\end{equation*}
	It follows from the proof of \cref{sec3 thm1} that $\{i : A^\perp_i \ne 0\}=\{0,q-p^m,q-1,q,q+1\}$. Since $p^m \ge 3$, we have $A^\perp_{q-2}=0$. Setting $j$ in \cref{sec1.1 equ1} to $q-1,q$ and $q+1$ in order, one obtains the desired weight enumerator of $\mathcal{C}^\perp_{(q,q+1,3,h)}$. Since $A^\perp_{q-2}=0$, it follows from the Assmus-Mattson Theorem that the minimum weight codewords in $\mathcal{C}^\perp_{(q,q+1,3,h)}$ support a $3$-$(q+1,q-p^m,\lambda)$ simple design with
	\begin{equation*}
		\lambda=\frac{(q-p^m)(q-p^m-1)(q-p^m-2)}{(p^{2m}-1)p^{m}},
	\end{equation*}
	where $\lambda$ is obtained from \cref{sec1.2 equ2,sec1.2 equ1}. The proof of the case of $h=\frac{q-p^i}{2}$ is completed.
	
	The proof of the case of $h=\frac{p^i-1}{2}$ can be established by considering the parity-check matrix
	\begin{equation}\label{sec4 thm1 equ2}
		\left.H_2=\left[\begin{array}{lllll}1&\beta^{-(h+1)}&(\beta^{-(h+1)})^2&\cdots&(\beta^{-(h+1)})^q\\1&\beta^{-h}&(\beta^{-h})^2&\cdots&(\beta^{-h})^q\\1&\beta^{h}&(\beta^{h})^2&\cdots&(\beta^{h})^q\\1&\beta^{h+1}&(\beta^{h+1})^2&\cdots&(\beta^{h+1})^q\end{array}\right.\right].
	\end{equation}
\end{proof}

\begin{remark}
	The AMDS codes and $3$-designs with the same parameters as in \cref{sec4 thm1} can be found in \cite[Theorem 5]{Xu2022InfiniteFamiliedDesigns}. The $3$-designs presented in \cref{sec4 thm1} were obtained from BCH codes, whereas the $3$-designs described in \cite{Xu2022InfiniteFamiliedDesigns} were based on linear codes generated by a special matrix over $\mathbb{F}_q$. With the help of Magma programs, the $3$-designs in \cref{sec4 thm1} are isomorphic to the $3$-designs in \cite[Theorem 5]{Xu2022InfiniteFamiliedDesigns} when $(p,s) = (2,4),(3,2),(3,3),(5,2)$. This prompts us to conjecture that the $3$-designs in \cref{sec4 thm1} are isomorphic to the $3$-designs in \cite[Theorem 5]{Xu2022InfiniteFamiliedDesigns} for any $(p,m)$.
\end{remark}

One can obtain the following theorem from the proof of \cref{sec4 thm1}.
\begin{theorem}
	Let $q=p^s$ where $s>1$ is an integer. Let $h=\frac{q-p^i}{2}$ for any $p$ or $h=\frac{p^i-1}{2}$ for $p$ odd, where $0<i<s$ is an integer. Denote by $\binom{U_{q+1}}{4}$ the set of all 4-subsets of $U_{q+1}$. Define 
	\begin{equation*}
		\mathcal{B}=\left\{\{x,y,z,w\}\in\binom{U_{q+1}}{4}:M(x,y,z,w)=0\right\},
	\end{equation*}
	where $M(x,y,z,w)$ is defined as \cref{sec4 thm1 equ1}. Suppose that $gcd(i,s)=m$. Then $(U_{q+1},\mathcal{B})$ is a $3$-$(q+1,4,p^m-2)$ simple design, and is isomorphic to the $3$-design supported by the minimum weight codewords of $\mathcal{C}_{(q,q+1,3,h)}$. Further, when $p$ is odd, the $3$-$(q+1,4,p^m-2)$ designs supported by the minimum weight codewords of $\mathcal{C}_{(q,q+1,3,\frac{q-p^i}{2})}$ and $\mathcal{C}_{(q,q+1,3,\frac{p^i-1}{2})}$ are isomorphic for the same $i$.
\end{theorem}

Combining \cref{sec3 cor2,sec4 thm1}, we obtain an infinite families NMDS codes holding $3$-designs when $p=3$. 
\begin{theorem}
	\label{sec4 thm2}
	Let $q=3^s$ where $s>1$ is an integer. Let $h=\frac{q-3^i}{2}$ or $h=\frac{3^i-1}{2}$, where $0<i<s$ is an integer. If $\gcd(i,s)=m=1$, then the BCH code $\mathcal{C}_{(q,q+1,3,h)}$ is a $[q+1,q-3,4]$ NMDS code with weight enumerator
	\begin{equation*}
		1+\frac{(q-1)^2q(q+1)}{24}z^{q-3}+\frac{(q^2-1)q(q+3)}{4}z^{q-1}+\\\frac{(q^2-1)(q^2-q+3)}3z^q+\frac{3(q-1)^2q(q+1)}8z^{q+1}.
	\end{equation*}
	Moreover, the minimum weight codewords in $\mathcal{C}^\perp_{(q,q+1,3,h)}$ support a $3$-$(q+1,q-3,\lambda)$ simple design with
	\begin{equation*}
		\lambda=\frac{(q-3)(q-4)(q-5)}{24}, 
	\end{equation*}
	and the minimum weight codewords in $\mathcal{C}_{(q,q+1,3,h)}$ support a $3$-$(q+1,4,1)$ simple design, i.e., a Steiner quadruple system $\mathrm{S}(3,4,q+1)$. Furthermore, the codewords of weight 5 in $\mathcal{C}_{(q,q+1,3,h)}$ support a $3$-$(q+1,5,(q-3)(q-7)/2)$ simple design.
\end{theorem}
\begin{proof}
	According to \cref{sec4 thm1}, we only need to prove the codewords of weight 5 in $\mathcal{C}_{(q,q+1,3,h)}$ support a $3$-$(q+1,5,(q-3)(q-7)/2)$ simple design. Since the minimum distance of $\mathcal{C}^\perp_{(q,q+1,3,h)}$ is $q-3$, $(\mathcal{P}(\mathcal{C}^\perp_{(q,q+1,3,h)}),\mathcal{B}_{q-4}(\mathcal{C}^\perp_{(q,q+1,3,h)}))$ is the trivial 3-design $(\mathcal{P}(\mathcal{C}^\perp_{(q,q+1,3,h)}),\emptyset)$. Note that $(\mathcal{P}(\mathcal{C}_{(q,q+1,3,h)}),\mathcal{B}_4(\mathcal{C}_{(q,q+1,3,h)}))$  is a 3-design. By setting $(r,t,\ell,\ell^\perp)=(q-6,3,4,q-4)$ in \cref{sec1 thm2}, one can deduced that the codewords of weight 5 in $\mathcal{C}_{(q,q+1,3,h)}$ support a $3$-design denoted by $(\mathcal{P}(\mathcal{C}_{(q,q+1,3,h)}),\mathcal{B}_5(\mathcal{C}_{(q,q+1,3,h)}))$. 
	
	Next we will show that $(\mathcal{P}(\mathcal{C}_{(q,q+1,3,h)}),\mathcal{B}_5(\mathcal{C}_{(q,q+1,3,h)}))$ is simple. Since two different codewords of weight 5 in $\mathcal{C}_{(q,q+1,3,h)}$ may have the same support, we then determine the number of non-repeating blocks in $\mathcal{B}_5(\mathcal{C}_{(q,q+1,3,h)})$.
	
	Let $h=\frac{q-3^i}{2}$ and $x,y,z,w$ be four pairwise distinct elements in $U_{q+1}$. Proceeding as in the proof of \cref{sec4 thm1 equ1} and using the same notation as there, we arrive at the conclusion that there exists a unique element $r(x,y,z)\in U_{q+1}\setminus\{x,y,z\}$ such that $f(r(x,y,z))=0$ since $3^m+1=4$. Hence, the rank of $M_1$ is 4 if and only if $w\notin \{x,y,z\}$. For any pairwise distinct elements $x,y,z,w,v\in U_{q+1}$, define
	\begin{equation*}
		M_2=\left[\begin{array}{lllll}x^{h}&y^{h}&z^{h}&w^{h}&v^{h}\\x^{h+1}&y^{h+1}&z^{h+1}&w^{h+1}&v^{h+1}\\x^{q-h}&y^{q-h}&z^{q-h}&w^{q-h}&v^{q-h}\\x^{q-h+1}&y^{q-h+1}&z^{q-h+1}&w^{q-h+1}&v^{q-h+1}\end{array}\right].
	\end{equation*}
	Since either $w$ or $v$ is not equal to $r(x,y,z)$, $M_2$ has a four-order non-zero minor. That is, the rank of $M_2$ equals 4. Let $\mathbf{c}=(c_{0},c_{1},\ldots,c_{q})$ be a codeword of weight 5 in $\mathcal{C}_{(q,q+1,3,h)}$ with nonzero coordinates $\{i_{1},i_{2},i_{3},i_{4},i_5\}$. Note that $M_2$ is a submatrix of the pairty-check matrix $H_1$. Then there exist five pairwise distinct element $x,y,z,w,v \in U_{q+1}$ such that
	\begin{equation*}
		\left[\begin{array}{lllll}x^{h}&y^{h}&z^{h}&w^{h}&v^{h}\\x^{h+1}&y^{h+1}&z^{h+1}&w^{h+1}&v^{h+1}\\x^{q-h}&y^{q-h}&z^{q-h}&w^{q-h}&v^{q-h}\\x^{q-h+1}&y^{q-h+1}&z^{q-h+1}&w^{q-h+1}&v^{q-h+1}\end{array}\right]\left.\left[\begin{array}{c}c_{i_1}\\c_{i_2}\\c_{i_3}\\c_{i_4}\\c_{i_5}\end{array}\right.\right]=\mathbf{0}.
	\end{equation*}
	Because the rank of $M_2$ is 4, the set $\{a\mathbf{c}:a\in \mathbb{F}_{q}^\ast\}$ consists of all such codewords of weight 5 in $\mathcal{C}_{(q,q+1,3,h)}$ with nonzero coordinates $\{i_{1},i_{2},i_{3},i_{4},i_5\}$. That is to say, if $B=\{i_{1},i_{2},i_{3},i_{4},i_5\}$, then $\{\mathbf{c}'\in\mathcal{C}_{(q,q+1,3,h)} : \mathrm{wt}(\mathbf{c}')=5, \mathrm{suppt}(\mathbf{c}')=B\}=\{a\mathbf{c}:a\in \mathbb{F}_{q}^\ast\}$. By the proof of \cref{sec4 thm1}, $A_4=\frac{(q-1)^2q(q+1)}{24}$. Setting $j=q-4$ in \cref{sec1.1 equ1} yields
	\begin{equation*}
		A_{5}=\binom{q+1}{5}(q-1)-(q-3)A_{4}=\frac{(q-7)(q-3)(q-1)^2q(q+1)}{5!}.
	\end{equation*}
	Hence, the number of non-repeating blocks in $\mathcal{B}_5(\mathcal{C}_{(q,q+1,3,h)})$ is $\frac{A_5}{q-1}$. In accordance with \cref{sec1.2 equ1}, the codewords of weight 5 in $\mathcal{C}_{(q,q+1,3,h)}$ support a $3$-$(q+1,5,(q-3)(q-7)/2)$ simple design. 
	
	The proof of the case of $h=\frac{3^i-1}{2}$ can be established by considering the parity-check matrix $H_2$ defined as \cref{sec4 thm1 equ2}. We omit it as it is very similar to the above proof.
\end{proof}

\begin{remark}\label{sec4 rmk1}
	Ding and Tang prove that the narrow-sense BCH code $\mathcal{C}_{(3^s,3^s+1,3,1)}$ over $\mathbb{F}_{3^s}$ and its dual are NMDS and hold $3$-designs in \cite[Theorem 21]{Ding2020InfiniteFamiliesMDS}. Actually, one can deduce that the result of \cite[Theorem 21]{Ding2020InfiniteFamiliesMDS} is a corollary of \cref{sec4 thm2} by taking $i=1$ and $h=\frac{3^i-1}{2}$ in \cref{sec4 thm2}.
\end{remark}

\section{Subfield subcodes of the MDS and NMDS codes}

Let $\mathcal{C}$ be an $[n,k,d]$ code over $\mathbb{F}_{r^t}$, where $r$ is a prime power. The {\it subfield subcode} $\mathcal{C}|_{\mathbb{F}_r}$ of $\mathcal{C}$ with respect to $\mathbb{F}_r$ is the set of codewords in $\mathcal{C}$ each of whose components is in $\mathbb{F}_r$, i.e. $\mathcal{C}|_{\mathbb{F}_r}=\mathcal{C}\cap\mathbb{F}_r^n$. This is an $[n,k^\prime,d^\prime]$ code with $n-t(n-k)\le k^\prime\le k$ and $d\le d^\prime$ \cite{Delsarte1975SubfieldSubcodesModified}. In this section, we present the subfield subcodes of the codes described in \cref{sec3 cor1,sec3 cor2}.

Let $q=p^s$, where $p$ is a prime and $s$ is an integer. We recall from the proof of \cref{sec3 thm1} that $\beta$ is a primitive $(q+1)$-th root of unity in $\mathbb{F}_{q^2}$ and 
\begin{equation}\label{sec5 equ1}
	\mathcal{C}_{(q,q+1,3,h)}^\perp=\left\{\left(\mathrm{Tr}_{q^{2}/q}\left(a\beta^{hi}+b\beta^{(h+1)i}\right)\right)_{i=0}^{q}:a,b\in\mathbb{F}_{q^2}\right\}.
\end{equation}
Let $t$ be a factor of $s$. Now we consider the subfield subcode $\mathcal{C}_{(q,q+1,3,h)}|_{\mathbb{F}_{p^t}}$. It follows from Delsarte’s theorem \cite[Theorem 2]{Delsarte1975SubfieldSubcodesModified} that
\begin{equation}\label{sec5 equ2}
	\mathcal{C}_{(q,q+1,3,h)}|_{\mathbb{F}_{p^t}}=\left(\mathrm{Tr}_{q/{p^t}}\left(\mathcal{C}_{(q,q+1,3,h)}^\perp\right)\right)^\perp.
\end{equation}
Combining \cref{sec5 equ1,sec5 equ2} leads to
\begin{equation*}
	\mathcal{C}_{(q,q+1,3,h)}|_{\mathbb{F}_{p^t}}=\left(\left\{\left(\mathrm{Tr}_{q^{2}/{p^t}}\left(a\beta^{hi}+b\beta^{(h+1)i}\right)\right)_{i=0}^{q}:a,b\in\mathbb{F}_{q^2}\right\}\right)^\perp.
\end{equation*}
By applying the Delsarte Theorem again, we have
\begin{equation*}
	\mathcal{C}_{(p^t,q+1,3,h)}^\perp=\left\{\left(\mathrm{Tr}_{q^{2}/{p^t}}\left(a\beta^{hi}+b\beta^{(h+1)i}\right)\right)_{i=0}^{q}:a,b\in\mathbb{F}_{q^2}\right\},
\end{equation*}
and so,
\begin{equation}\label{sec5 equ3}
	\mathcal{C}_{(q,q+1,3,h)}|_{\mathbb{F}_{p^t}}=\mathcal{C}_{(p^t,q+1,3,h)}.
\end{equation}

Let $q=2^s$. Taking $i=s-1$ in \cref{sec3 cor1}, one can deduce that the BCH code $\mathcal{C}_{(q,q+1,3,\frac{q}{4})}$ is an MDS code with parameters $[q+1,q-3,5]$. We then consider the parameters of subfield subcodes $\mathcal{C}_{(q,q+1,3,\frac{q}{4})}|_{\mathbb{F}_{2}}$ and  $\mathcal{C}_{(q,q+1,3,\frac{q}{4})}|_{\mathbb{F}_{4}}$. 

\begin{theorem}\label{sec5 thm1}
	Let $q=2^s$ and $h=\frac{q}{4}$, where $s\ge 4$ is an integer. Then the binary subfield subcode $\mathcal{C}_{(q,q+1,3,h)}|_{\mathbb{F}_{2}}$ has parameters $[q+1,q+1-4s,d\ge5]$.
\end{theorem}
\begin{proof}
	According to \cref{sec5 equ3}, $\mathcal{C}_{(q,q+1,3,h)}|_{\mathbb{F}_{2}}=\mathcal{C}_{(2,q+1,3,h)}$. Let $n=q+1$. The 2-cyclotomic coset $C_1$ modulo $n$ is given by
	\begin{equation*}
		C_1=\{1,2,\ldots,2^{s-1},-1,-2,\ldots,-2^{s-1}\}\bmod n.
	\end{equation*}
	Since $h=\frac{q}{4}=2^{s-2}$,
	\begin{equation*}
		C_h=2^{s-2}C_1=\{2^{s-2},2^{s-1},-1,-2,\dots,-2^{s-1},1,2,\ldots,2^{s-3}\}=C_1\bmod n.
	\end{equation*}
	Clearly, $|C_h|=|C_1|=2s$. Note that $h+1 \equiv 2^{s-2}+1 \equiv 2^{s-2}-2^s \equiv -3\cdot2^{s-2} \bmod n$. Then
	\begin{equation*}
		C_{h+1}=(2^{s-2}+1)C_1=(-3\cdot2^{s-2})C_1=3C_1=C_3\bmod n.
	\end{equation*}
	It is easily checked that $C_h \cap C_{h+1} =\emptyset$ when $s\ge3$. 
	
	Next we will show that $|C_{h+1}|=2s$ when $s\ge4$. Assume on the contrary that the size $\ell_{3}$ of $C_{3}$ satisfies $1\le \ell_{3}\le 2s-1$. By \cref{sec2.1 thm1}, $\ell_3$ is a proper divisor of $2s$. Then $1\le \ell_3 \le s$. If $\ell_3=s$, $3\equiv-3\bmod 2^s+1$, which is impossible. If $1\le \ell_{3}\le s-1$, then $2^s+1\le 3(2^{\ell_{3}}-1)$, and so $\ell_{3}=s-1$. Thus, $s=3$, which gives a contradiction to $s\ge 4$. Therefore, $|C_{h+1}|=2s$ when $s\ge4$.
	
	By \cref{sec2.1 equ3}, the minimal polynomial $M_{\beta^j}(x)$ of $\beta^j$ over $\mathbb{F}_2$ is given by
	\begin{equation*}
		\mathsf M_{\beta^j}(x)=\sum_{i\in C_j}(x-\beta^i)
	\end{equation*}
	for $j\in \{h, h+1\}$. Since $C_h \cap C_{h+1} =\emptyset$, the generator polynomial of $\mathcal{C}_{(2,q+1,3,h)}$ is $M_{\beta^h}(x)M_{\beta^{h+1}}(x)$. Hence, the dimension of $\mathcal{C}_{(q,q+1,3,h)}|_{\mathbb{F}_{2}}$ is $q+1-4s$. Note that the minimum distance of the code $\mathcal{C}_{(q,q+1,3,h)}$ is 5. By the definition of subfield subcode, the minimum distance $d$ of $\mathcal{C}_{(q,q+1,3,h)}|_{\mathbb{F}_{2}}$ is not less than 5. The proof is completed.
\end{proof}

\begin{example}	
	Let $\mathcal{C}_{(2^s,2^s+1,3,2^{s-2})}|_{\mathbb{F}_{2}}$ be the subfield subcode in \cref{sec5 thm1}.
	\begin{enumerate}[label=\arabic*)]
		\item When $s=4$, $\mathcal{C}_{(2^s,2^s+1,3,2^{s-2})}|_{\mathbb{F}_{2}}$ is an MDS code with parameters $[2^s+1,1,2^s+1]$. However, it is trivial as its dimension is 1.
		\item When $s=5$, $\mathcal{C}_{(2^s,2^s+1,3,2^{s-2})}|_{\mathbb{F}_{2}}$ is a $[33,13,10]$ binary linear code, and its dual has parameters $[33,20,6]$. Both codes have the best known parameters in accordance with \cite{Grassl}.
		\item When $s=6$, $\mathcal{C}_{(2^s,2^s+1,3,2^{s-2})}|_{\mathbb{F}_{2}}$ is a $[65,41,5]$ binary linear code, and its dual has parameters $[65,24,16]$. However, the corresponding best-known parameters are $[65,41,8]$ and $[65,24,17]$ according to \cite{Grassl}.
	\end{enumerate}
\end{example}

The proofs of \cref{sec5 thm2,sec5 thm3} are similar to the proof of \cref{sec5 thm1}, therefore we abbreviate them. 

\begin{theorem}\label{sec5 thm2}
	Let $q=2^s$ and $h=\frac{q}{4}$, where $s\ge 2$ is an even integer. Then the quaternary subfield subcode $\mathcal{C}_{(q,q+1,3,h)}|_{\mathbb{F}_{4}}$ has parameters $[q+1,q+1-2s,d\ge5]$.
\end{theorem}
\begin{proof}
	Since $C_h \cap C_{h+1} =\emptyset$ and $|C_{h}|=|C_{h+1}|=s$, the dimension of $\mathcal{C}_{(q,q+1,3,h)}|_{\mathbb{F}_{4}}$ is $q+1-2s$.  This completes the proof.
\end{proof}

\begin{example}
	Let $\mathcal{C}_{(2^s,2^s+1,3,2^{s-2})}|_{\mathbb{F}_{4}}$ be the subfield subcode in \cref{sec5 thm2}.
	\begin{enumerate}[label=\arabic*)]
		\item When $s=2$, $\mathcal{C}_{(2^s,2^s+1,3,2^{s-2})}|_{\mathbb{F}_{4}}$ is an MDS code with parameters $[5,1,5]$, which is trivial.
		\item When $s=4$, $\mathcal{C}_{(2^s,2^s+1,3,2^{s-2})}|_{\mathbb{F}_{4}}$ is a $[17,9,7]$ quaternary linear code, and its dual has parameters $[17,8,8]$. Both $\mathcal{C}_{(2^s,2^s+1,3,2^{s-2})}|_{\mathbb{F}_{4}}$ and its dual have the best known parameters in accordance with \cite{Grassl}.
		\item When $s=6$, $\mathcal{C}_{(2^s,2^s+1,3,2^{s-2})}|_{\mathbb{F}_{4}}$ is a $[65,53,5]$ quaternary linear code, while the corresponding best known parameters are $[65,53,6]$ according to \cite{Grassl}. Its dual is a quaternary code with parameters $[65,12,32]$ which is the best known in accordance with \cite{Grassl}.
	\end{enumerate}
\end{example}

Let $q=3^s$. Taking $i=1$ and $h=\frac{q-3^i}{2}$ in \cref{sec3 cor2}, one can deduce that the BCH code $\mathcal{C}_{(q,q+1,3,\frac{q}{4})}$ is an NMDS code with parameters $[q+1,q-3,4]$. The following theorem provides the information of the subfield subcode $\mathcal{C}_{(q,q+1,3,\frac{q-3}{2})}|_{\mathbb{F}_{3}}$. 

\begin{theorem}\label{sec5 thm3}
	Let $q=3^s$ and $h=\frac{q-3}{2}$, where $s\ge 2$ is an integer. Then the ternary subfield subcode $\mathcal{C}_{(q,q+1,3,h)}|_{\mathbb{F}_{3}}$ has parameters $[q+1,q+1-4s,d\ge4]$.
\end{theorem}
\begin{proof}
	One can check that $C_h \cap C_{h+1} =\emptyset$ and $|C_h|=|C_{h+1}|=2s$, which implies that the dimension of $\mathcal{C}_{(q,q+1,3,h)}|_{\mathbb{F}_{3}}$ is $q+1-4s$. The proof is completed.
\end{proof}

\begin{example}
	Let $\mathcal{C}_{(3^s,3^s+1,3,\frac{3^s-3}{2})}|_{\mathbb{F}_3}$ be the subfield subcode in \cref{sec5 thm2}. In the following, we present the pamaraters of $\mathcal{C}_{(3^s,3^s+1,3,\frac{3^s-3}{2})}|_{\mathbb{F}_3}$ and its dual for $s=2,3,4$.
	\begin{equation*}
		\begin{array}{ccc}
			s&\mathcal{C}_{(3^s,3^s+1,3,\frac{3^s-3}{2})}|_{\mathbb{F}_3}&(\mathcal{C}_{(3^s,3^s+1,3,\frac{3^s-3}{2})}|_{\mathbb{F}_3})^\perp\\2&[10,2,5]&[10,8,2]\\3&[28,16,4]&[28,12,8]\\4&[82,66,6]&[82,16,36]
		\end{array}
	\end{equation*}
	When $s = 2$ and $s = 3$, both $\mathcal{C}_{(3^s,3^s+1,3,\frac{3^s-3}{2})}|_{\mathbb{F}_3}$ and its dual are the best known ternary cyclic codes according to \cite[Appendix A]{Ding2014CodesDifferenceSetsa}. In particular, $(\mathcal{C}_{(3^s,3^s+1,3,\frac{3^s-3}{2})}|_{\mathbb{F}_3})^\perp$ has the best known parameters when $s=2$ in accordance with \cite{Grassl}.
\end{example}

\section{Concluding remarks}\label{sec6}

By exploring the parameters of dual codes of two families of BCH codes, we discovered an infinite family of MDS codes over $\mathbb{F}_{2^s}$, and two infinite families of AMDS codes over $\mathbb{F}_{p^s}$ for any prime power $p^s$, which include two infinite families of NMDS codes over $\mathbb{F}_{3^s}$. Notably, \cref{sec1 conj1} has been addressed through a subclass of these NMDS codes. Additionally, we provided several binary, ternary, and quaternary codes with the best-known parameters. According to the Assmus-Mattson theorem, we established that the minimum weight codewords of these AMDS codes and their duals support infinite families of $3$-designs. Moreover, we proved that the codewords of weight 5 in these NMDS codes also support 3-designs, which do not satisfy the Assmus-Mattson theorem. Notably, our findings regarding NMDS codes over \(\mathbb{F}_{3^s}\) holding $3$-designs encompassed the relevant results in \cite[Theorem 21]{Ding2020InfiniteFamiliesMDS}. 

In the proof of \cref{sec3 thm1,sec3 thm4}, we analyzed the solutions in $U_{q+1}$ of certain special equations over $\mathbb{F}_{q^2}$. Similarly, infinite families of AMDS and NMDS codes supporting $t$-designs were obtained by determining the solutions of certain special equations over $\mathbb{F}_q$ in \cite{Xu2022InfiniteFamiliedDesigns,Heng2023NewInfiniteFamilies}. Therefore, it would be very interesting to explore the possibility of obtaining infinite families of AMDS or NMDS codes that support $t$-designs through the analysis of some special equations over finite field.

%\backmatter
%
%\bmhead{Acknowledgements}
%
%Acknowledgements are not compulsory. Where included they should be brief. Grant or contribution numbers may be acknowledged.
%
%Please refer to Journal-level guidance for any specific requirements.
%
%\section*{Declarations}
%
%Some journals require declarations to be submitted in a standardised format. Please check the Instructions for Authors of the journal to which you are submitting to see if you need to complete this section. If yes, your manuscript must contain the following sections under the heading `Declarations':



\end{document}